\documentclass[12pt]{article}
\usepackage{color}

\usepackage{amsfonts,amsmath,amsxtra,amsthm,amssymb,latexsym}
\usepackage{amssymb}  
\usepackage{bbm} 
\usepackage{graphicx}

 \usepackage[utf8]{inputenc}
 \usepackage{multicol}

\newtheorem{thm}{Theorem}[section]
 \newtheorem{cor}[thm]{Corollary}

  \theoremstyle{definition}
 \newtheorem{defn}{Definition}[section]
 \theoremstyle{remark}
 \newtheorem{rem}{Remark}[section]
 \newtheorem{ex}[thm]{Example}

 \numberwithin{equation}{section}

 \let\emptyset\varnothing

\def\RR{\mathbb R} \def\CC{\mathbb C} \def\NN{\mathbb N}
\def\ZZ{\mathbb Z} \def\QQ{\mathbb Q}

\def\DD{\mathbb D}  \def\BB{\mathbb B}
\def\PP{\mathbb P} \def\EE{\mathbb E} 
\def\VV{\mathbb{V}}

\def\ii{\mathrm{i}} \def\dd{\mathrm{d}} \def\ee{\mathrm{e}}

\def\wt{\widetilde}  \def\pa{\partial}

\DeclareMathOperator{\im}{Im}
\DeclareMathOperator{\re}{Re}

\DeclareMathOperator{\Ln}{Ln}
\DeclareMathOperator{\mult}{mult}
\DeclareMathOperator{\diam}{diam}
\DeclareMathOperator{\Arg}{Arg}

\DeclareMathOperator{\Var}{Var}

\def\al{\alpha}  

\def\la{\lambda}\def\La{\Lambda}

\def\de{\delta} \def\De{\Delta}
\def\ep{\epsilon} 
 \def\Ga{\Gamma}
\def\om{\omega} \def\Om{\Omega}
 
\def\Si{\Sigma} \def\si{\sigma}
\def\Up{\Upsilon}

\def\N{\mathfrak{N}}

\def\B{\beta}
\def\V{V}

\def\K{\mathcal{K}}

\def\F{\mathcal{F}}

\def\AAA{\mathbb{A}}

\def\muren{\mathfrak{m}}

\def\Ad{\mathrm{Ad}}

\def\Kp{K}

\def\Bd{\partial}

\def\comp{\mathrm{comp}}

\def\id{\mathfrak{e}}

\begin{document}

\title{Asymptotics of random resonances generated by a point process of delta-interactions}
\author{}
\date{}
\maketitle

{\center  {\large 
Sergio Albeverio $^{\text{a}}$ and Illya M. Karabash $^{\text{b,c,d}}$ \\[6ex]
}}

\hspace{0.3\linewidth}\parbox{0.65\linewidth}{
\small \it
This paper is dedicated to the dear memory of Erik Balslev. 
Erik has done groundbreaking seminal work on the relations between Schrödinger operators, complex scaling, spectral theory, and number theory. The first named author has had the great pleasure to meet him in Princeton in 1970 and received much inspiration from him. Erik has been a very kind, deep,  and open person, and the very dear friend.}

\vspace{4ex}

{\small \noindent
$^{\text{a}}$  Institute for Applied Mathematics, Rheinische Friedrich-Wilhelms Universität Bonn,
and Hausdorff Center for Mathematics, Endenicher Allee 60,
D-53115 Bonn, Germany\\[1mm]
$^{\text{b}}$ Fakultät für Mathematik, TU Dortmund, Vogelpothsweg 87, 44227 Dortmund, Germany\\[1mm]
$^{\text{c}}$ Institute of Applied Mathematics and Mechanics of NAS of Ukraine, Dobrovolskogo st. 1, Slovyans'k 84100, Ukraine\\[1mm]
$^{\text{d}}$ Corresponding author: i.m.karabash@gmail.com\\[2mm]
E-mails: albeverio@iam.uni-bonn.de, i.m.karabash@gmail.com
}

\begin{abstract}
We introduce and study the following model for random resonances: we take a  collection of point interactions $\Up_j$ generated by a simple finite point process in $\RR^3$ and consider the resonances of associated random Schrödinger Hamiltonians 
$H_\Up = -\De + ``\sum \muren (\al) \de (x - \Up_j)``$. These resonances are zeroes of a random exponential polynomial, and so form a point process $\Si (H_\Up)$ in the complex plane  $\CC$.
We show that the counting function for the set of random resonances $\Si (H_\Up)$ in $\CC$-discs with growing radii possesses Weyl-type asymptotics almost surely for a uniform binomial process $\Up$, and obtain an explicit formula for the limiting distribution as $m \to \infty$ of the leading parameter of the asymptotic chain of `most narrow' resonances generated by a sequence of uniform binomial processes $\Up^m$ with $m$   points. We also pose a general question about limiting behavior of the point process formed by leading parameters of asymptotic sequences of resonances.   Our study leads to questions about metric characteristics for the combinatorial geometry of $m$ samples of a random point in the 3-D space and related statistics of extreme values. 
\end{abstract}

{ \small \noindent
MSC-classes: (Primary) 
82B44  
35B34, 
35P20, 
60G70; 
(Secondary)
60G55,  
35P25, 
35J10, 
60H25, 
47B80, 
81Q80 

\tableofcontents
}

\normalsize

\section{Introduction}
\label{s:i}

This paper is written in Memory of Erik Balslev. Erik's pioneering works have been very influential in many areas of analysis, mathematical physics, and  analytical number theory. Some of related topics are closely connected with the present paper.  In fact, one of Erik's major contributions into mathematical physics is the clarification of the very concept of resonance \cite{B87,BS88,BS89,B89_I,B89_II}. Moreover, the technique of complex dilations that he developed in cooperation with Jean-Michel Combes \cite{BC71} has played an important role in making accessible tools of analytic perturbation theory in problems involving resonances and eigenvalues embedded in continuous spectra. The connection of spectral theory with problems of analytic number theory, that Erik discovered and masterly developed in a series of important papers \cite{BV98,BV01,BV06}, has inspired us to the study of the interplay between resonances, exponential polynomials, and point interactions.

In this paper we introduce a 3-D continuous model for random resonances using Hamiltonians of a generalized Schrödinger type involving point interactions. To make the Hamiltonian random, we assume that interactions are generated by a point processes with a suitable properties. The main goal of this paper is to introduce main notions and problems for related random resonances and consider some of their asymptotic properties on a relatively simple examples of binomial point processes. 

Another part of our research aimed on a more detailed study of the asymptotics of random resonances is now in preparation for publication.

While Schrödinger operators with random point interactions have been introduced in mathematical papers  and their self-adjoint spectra have been investigated (see  \cite{KM82,FHT88,AGHH12,KMN19} and references therein), it seems that the resonances for such models are not adequately mathematically studied. 

For other models of random resonances,  the mathematical theory has been attracting an increasing attention during recent years. It worth to mention the monograph \cite{S14} and the paper \cite{K16}. One of the problems suggested in the introduction to \cite{S14} as a promising direction of future research concerns the connection between Weyl asymptotics and asymptotics of random resonances. We would like to note that Section \ref{s:ad} of the present paper addresses a somewhat connected problem in the context of Schrödinger operator with random point interactions.
Namely, we prove that the Weyl-type asymptotics, which has been recently introduced for deterministic point interactions in \cite{LL17}, takes place almost surely (a.s.) for our stochastic example. 

From a more general perspective, random resonance effects were intensively studied in Physics (see the literature in \cite{K16}).
One of the first mentioning in the mathematical context of resonances of random Schrödinger operator known to us is in  the paper \cite{HS86}, where 
the question of estimation of the support of distribution of random resonances served as one of motivations for a resonance optimization problem (see also \cite{KLV17} for a more recent discussion of this interplay).

We shall present now in detail the model of random resonances that will be studied in this paper. Let $\Up$ be a  point process on $\RR^3$, let $(\Om,\F,\PP)$ be the underlying complete probability space, and let $\eta_{_\Up}$ be the random (counting) measure associated with $\Up$. Throughout the paper we assume that
\begin{itemize}
\item[(A0)] the point process $\Up$ is simple and finite. 
\end{itemize}
(A point process $\La$ on $\RR^d$ is said to be finite if it satisfies $\eta_{_\La} (\RR^3) < \infty $ almost surely; for the definition of simple process see below).

Any locally-finite point process on $\RR^d$ is proper. For the finite point process $\Up$, this means that there exist random variables $\nu:\Om \to  \NN_0 = \{0\} \cup \NN $ and 
$\Up_j : \Om \to \RR^3$, $j \in \NN$, such that $\Up $ can be considered as a finite collection of random points $\Up (\om)= \{ \Up_j (\om)\}_{j=1}^\nu$ for almost all (a.a.) $\om \in \Om$ (with respect to the measure $\PP$).

In particular,  $\eta_{_\Up} = \sum_{j=1}^{\# \Up} \de (\cdot -\Up_j (\om)) \dd x$ a.s., where $\de (x) \dd x$ is Dirac's delta measure and $\# S$ is the number of elements in a set (or in a multiset) $S$. 
When $j_2<j_1$ or $\nu =0$, it is assumed that $ \sum_{j=j_2}^{j_1} = 0$ or, resp., $\{ \Up_j \}_{j=1}^0 = \emptyset$. For basic definitions and facts concerning point processes, we refer to \cite{LP17} (see also \cite{K06,H12}).

We associate with $\Up$ a random Hamiltonian $H_\Up$ in the following way. Let $\al $ be a complex number, which is fixed throughout the paper.
For a deterministic set $Y =\{y_j\}_{j=1}^{\# Y} $ consisting of $\# Y \in \NN$ distinct points 
in $\RR^3$, the linear operator $H_Y$ in the complex Hilbert space $L^2_\CC (\RR^3)$ is the point interaction Hamiltonian 
corresponding to the formal differential expression 
\begin{equation}\label{e:H}
-\De u (x) + `` \sum_{j=1}^{\# Y} \muren (\al) \de (x - Y_j) u (x)  ``, \quad x \in \RR^3 , 
\end{equation}
and the `strength-type' parameter $\al$. Here the Dirac measure $\delta (\cdot - y_j)$ placed at a \emph{center} $y_j \in \RR^3$ of a \emph{point interaction} is symbolically multiplied on 
a normalization parameter $\muren (\al)$, see \cite{AH84,AGHH12,AK17} for details and 
Section \ref{s:PP} for the rigorous definition of this deterministic operator.
If $\#Y=0$, we assume that $Y = \varnothing$ and $H_Y = - \De$, where $\De= \sum_{j=1}^3 \pa^2_{x_j}$ is the Laplacian operator 
in $\RR^3$. 

The aim of this paper is to study resonances   of the random operator $H_\Up$.
For a deterministic Hamiltonian $H$, (continuation) \emph{resonances} $k$ are defined as poles of the resolvent $(H-k^2)^{-1}$ extended in a generalized sense through essential spectrum  into the  lower complex half-plane $\CC_- := \{ z \in \CC : \im z <0\}$ \cite{AH84,M84}.
The collection of all resonances $\Si (H) \subset \CC$ associated with $H$ (in short, resonances of $H$) is actually a \emph{multiset}, i.e., a set in which 
an element $e$ can be repeated a finite number of times.
We denote this number $\mult e$ and call it the multiplicity of $e$. 

An element $e$ is called multiple (simple) if  $\mult e \ge 2$ (resp., $\mult e = 1$). A multiset is said to be simple if every of its elements is simple. A point process is called simple if it is a.s. simple. 

The (algebraic) \emph{multiplicity of a resonance} $k$ can be  
defined as the multiplicity of the corresponding generalized pole of 
$(H-k^2)^{-1}$ (e.g., \cite{DZ19}), as the multiplicity of an eigenvalue obtained by a complex dilation \cite{BC71,SZ91,DZ19}, or  as the multiplicity of 
a zero of a certain analytic function built from the resolvent of $H$ 
and generating resonances as its zeros \cite{AH84}. In the present paper, we follow the latter approach  to the definition of multiplicity (see \cite{AH84,AGHH12,AK17} and Section \ref{s:PP}).

For $\om$ belonging to the event $\{ \om \in \Om: \eta_\Up (\RR^3) = \infty \}$ or to the event $\{\om \in \Om : \Up (\om)   \text{ is non-simple} \}$, which both have zero probability according to (A0), we do not define the Hamiltonian $H_{\Up (\om)}$. So, $\Si (H_\Up (\cdot ))$ is defined only almost surely.

\begin{rem} \label{r:def0} 
The reason for this convention is the following. While $H_Y$ can be defined in some cases when $\# Y = \infty$
and $Y$ is simple, e.g., in the case where $\inf_{j \neq j'} |Y_j - Y_{j'}|>0$ \cite{AGHH12} (see \cite{KMN19} for a relaxation of this condition in deterministic and stochastic settings), the notion of resonances in such settings is not understood well. It is natural, e.g., to expect that the resonances can be defined in a certain way for the case of a periodic lattice, but the classification of singularities of the resolvent and their physical meaning requires an additional study (cf. \cite{K97,K11}).  
We believe that in some of the cases with $Y$ having multiple points it is also possible to give some meaning to the operator $H_Y$ and, moreover, that this question is important for the optimization of resonances of $H_Y$ by a modification of the positions of centers (cf. \cite{K14,AK17}), but we are not aware about such studies.
\end{rem}


The random collection of resonances $\Si (H_\Up)$ is a proper point process in $\CC$ (generally, with multiplicities), see Theorem \ref{t:pp} below. In this paper, we are interested mainly not in the point process of resonance $\Si (H_\Up)$ itself, but in its random asymptotical behaviour near $\infty$ in the complex plane of the spectral parameter $k$. 

In Section \ref{s:ad}, we study this behaviour on the `rough level' of 
the asymptotics of counting function $\eta_{_{\Si (H_{_\Up})}} (\DD_R)$ as $\RR \to +\infty$,
and the (normalized) asymptotic density of resonances defined by
\[
\Ad (H_{\Up}) := \lim_{R \to \infty} \frac{\eta_{_{\Si (H_{_\Up})}} (\DD_R)}{R} \quad \text{ (see \cite{LL17,AK19,AK20})},
\]  
where $\DD_R = \{ z \in \CC  \ : \ |x| < R\}$.

The class $\Theta (m,\BB_r)$ of the binomial processes  of the sample size $m$ with a uniform sampling distribution in a ball $\BB_r := \{ x \in \RR^3 : |x| <r\}$ provides us with a simple  
example of a point process  $\Up$ with `good' diffuse properties (see, e.g., \cite{LP17}). We show in Theorem \ref{t:Ad=as} that 
the corresponding asymptotic density 
$
\Ad (H_{\Up}) 
$
is equal a.s. to the random variable $\frac{V (\Up)}{\pi}$, where 
\[
V (Y) := \max_{\si \in S_N} \sum_{j=1}^{N} | Y_{j} -  Y_{\si (j)}| 
\]
is the `size' of a deterministic collection $Y$ of $N$ points (see \cite{LL17}) and $S_N$ is the symmetric group of degree $N$ consisting of permutations $\si$.
(The above notation considers a permutation as a bijective map $\si:\{1,\dots,N\} \to \{1,\dots,N\}$). 

Slightly modifying the terminology of \cite{LL17}, we say that a  deterministic operator $H_Y$ has the Weyl-type asymptotics of resonances if $\Ad (H_{Y}) = \frac{V (\Up)}{\pi}$. So, we proved that Weyl-type asymptotics holds a.s. for the case $\Up \in \Theta (m,\BB_r) $, but our proof can be easily extended to a wide class of binomial processes. 

In Section \ref{s:AsPp}, we consider the fine structure of asymptotical  behaviour of the resonance point process $\Si (H_\Up)$ near $k=\infty$. This structure can be seen with the use of the logarithmic (asymptotic)  density function $\Ad^{\log} (h)$ \cite{AK19}, i.e., the normalized density  corresponding to semi-logarithmic strips 
$\{ - h \ln (|\re z|+1) \le \im z \}$, which is defined by  
\begin{align}
& \Ad^{\log} (h)  = \Ad^{\log} (\Up,h) := \lim_{R\to \infty} 
\frac{\N^{\log}_{H_{\Up (\om)}} (h,R)}{R} ,
\end{align}
via the associated \emph{logarithmic counting function} 
\begin{align}
& \N^{\log}_{H_{\Up (\om)}} (h,R)  :=  \# \{ k \in \Si (H_{\Up (\om)}) : 
- h \ln (|\re k|+1) \le \im k \ \text{ and } \ |k| \le R \} .
\end{align}
It is easy to see that  $\Ad (h) $ is a random variable for each $h \in \RR$.

The main result of Section \ref{s:AsPp} says that the sample paths of $\Ad^{\log} (h)$, $h\in \RR$, are almost surely piecewise constant functions with finite number of jumps and that after multiplication by $\pi h$ the associated random measure $ \dd \Ad^{\log} (h)$ becomes $\NN_0$-valued and defines a finite point process $\K := \{\K_j \}_{j=1}^{\# \K}$ on $\RR_+$. Thus, this point process describes the structure of  asymptotic sequences of random resonances going to $\infty$ in the $k$-plane.

In Section \ref{s:limits}, we take a sequence of uniform binomial processes $ \Up^{[m]} \in  \Theta (m,\BB_{r})$ in a fixed 3-D ball $\BB_r$ and consider the limiting behavior of the process $\K$ when the number of points $m$ grows to $\infty$. In particular, we obtain an explicit formula for the normalized limiting distribution of 
\[
\min \K (\Up^{[m]}) = \min_{1\le j \le \, \#\K(\Up^{[m]})}  \K_j (\Up^{[m]})
\]
 as $m \to \infty$. We also obtain probabilistic  estimates for the distributions of the asymptotic densities $\Ad (H_{\Up^{[m]}})$ for large $m$ using the Weyl-type asymptotics result of Section \ref{s:ad}.

\textbf{Notation.}
For a set $Y \in \RR^d$, its diameter is defined by
\[
\diam Y := \inf \{ \ell \ge 0  \ : \ \ell \ge |y_0 - y_1| \text{ for all } y_0,y_1 \in Y \} \text{ (we want $\diam \emptyset := 0$)}.
\]
The following standard sets are used: 
the set $\ZZ$ of integers, the set $\NN_0 = \NN \cup \{0\}$ of nonnegative integers,  half-lines $\RR_\pm = \{ x \in \RR : \pm x>0\}$, 
discs $\DD_R = \{ z \in \CC  \ : \ |x| < R\}$ in the complex plane, open half-planes $\CC_\pm = \{ z \in \CC : \pm \im z >0 \}$, and balls $\BB_r := \{ x \in \RR^3 : |x| <r\}$ in the 3-D real space. Sometimes, the complex plane $\CC$ is considered as the  Euclidean space $\RR^2$. For a subset $S$ of a normed space $U$, we denote its closure by  $\overline{S}$ and by $\Bd S$ its boundary. For $u_0 \in U$ and $z \in \CC$, we write 
$
z S  + u_0 := \{ zu + u_0 \, : \, u \in S \}
$. 
The class of random vectors with the standard normal distribution in $\RR^3$ is denoted by $N (\mathbf{0},I_{\RR^3})$.

The function $\Ln (\cdot)$ 
is the branch of the natural logarithm multi-function $\ln (\cdot)$ 
in $\CC \setminus (-\infty,0]$ fixed by  $\Ln 1 = \ii \Arg_0 1 = 0$. 
For $z \in \RR_-$, we put $\Ln z = \Ln |z| + \ii \pi $.

\section{Point process of random resonances\label{s:PP}}

Let $\{a_n\}_{n \in \NN_0}$ be a sequence of $\CC$-valued random variables on the complete probability space $(\Om,\F,\PP)$. Then the random power series $F (z) = \sum_{n=0}^{\infty} a_n z^n$ is said to be a random entire function (on $\CC$) if the radius of convergence of $F$ a.s. equals $\infty$. The random entire function $F$ is said to be a random polynomial if a.s. $a_n \neq 0$ only for a finite number of indices $n \in \NN_0$ (we refer to \cite{A66_Rps,A66_JFRAM,BS86} for basic facts of the theories of random analytic functions and random polynomials).

In this section, it will be shown that the multiset $\Si (H_\Up)$ of resonances of the random operator $H_\Up$ is a.s. the multiset of zeros of a random entire function. Then it is easy to see that $\Si (H_\Up)$ is a point process in $\CC$.

\begin{defn}
Consider a finite simple point process $\B = \{ \B_j \}_{j=1}^{\# \B}$ on $\CC$. Let $\{p_j \}_{j=1}^{\infty}$ be a sequence of random polynomials. Then any  random function of the form 
\[
\sum_{j=1}^{\#\B} e^{\B_j z} p_j (z) 
\]
is a.s. defined on the whole complex plane $\CC$ and is said to be a random exponential polynomial.
\end{defn}

It is easy to check that a random exponential polynomial is a random entire function (concerning the theory of deterministic exponential polynomials we refer to \cite{BC63,BG12}).

In the rest of the paper we suppose that the assumption (A0) holds. Then from the definition of resonances for deterministic point interaction Hamiltonians $H_Y$ \cite{AH84}, one sees that the multiset of random resonances $\Si (H_\Up)$ is a.s. the multiset of zeros of a specially constructed random exponential polynomial. 

For convenience of the reader, let us recall the construction of this exponential polynomial and the definition of the operator $H_Y$ associated with (\ref{e:H}) in the deterministic settings.
Let $Y=\{ Y_j \}_{j =1}^{\#Y}$ be a simple finite collection of points in $\RR^3$, i.e.,
the interaction centers $Y_j$ are distinct and their number $\# Y \in \NN \cap \{0\}$ is finite.
We assume that all point interactions are of the same `strength' which is described by a fixed parameter $\al \in \CC$. Various approaches to the definitions of the  operator $H_Y$ associated with (\ref{e:H}) were given, e.g., in 
\cite{AFH79,AGH82,AGHS83,AGHH12,AK17}.
It is a closed operator in the complex Hilbert space $L^2_\CC (\RR^3)$ 
 and it has a nonempty resolvent set which can be obtained from $\CC \setminus [0,+\infty)$ after  possible exclusion of a finite number of points \cite{AGHH12,AK17}.

The resolvent $(H_Y-z^2)^{-1}$ of $H_Y$ is defined in the classical sense 
on the set of $ z \in \CC_+ $ such that $z^2$ is not in the spectrum. Its integral kernel has the form 
\begin{gather} \label{e:Res}
(H_Y-z^2)^{-1} (x,x')  = G_z (x-x') + \sum_{j,j' = 1}^{\#Y} G_z (x-Y_j)
\left[ \Ga_Y \right]_{j,j'}^{-1} G_z (x' - Y_{j'} ) , 
\end{gather}
where $x,x' \in \RR^3 \setminus Y$ and $x \neq x'  $, see e.g. \cite{AGHH12,AK17}.
Here 
$
G_z (x-x') := \frac{e^{\ii z |x-x'|}}{4 \pi |x-x'|}
$ 
is the integral kernel associated with
the resolvent $(-\De - z^2)^{-1}$ of 
the kinetic energy Hamiltonian $-\De$;
$\left[ \Ga_Y \right]_{j,j'}^{-1}$ denotes the $j,j'$-element of the inverse to 
the matrix 
\begin{gather} \label{e:Ga}
\Ga_Y (z) = \left[ \left( \al - \tfrac{\ii z}{4 \pi} \right) \de_{jj'} 
- \wt G_z  (Y_j-Y_{j'})\right]_{j,j'=1}^{\# Y}, \text{ where }
\wt G_z (x) := \left\{ \begin{array}{rr} G_z (x), & 
x \neq 0 \\
0 , & 
x = 0  \end{array} \right.  .
\end{gather}

\begin{rem}
Actually, the Krein-type  formula (\ref{e:Res}) for 
the difference of the perturbed and unperturbed resolvents of operators $H_Y$ and $-\De$ can be used as a definition of $H_Y$ \cite{GHM80,AGHH12}.
For other equivalent definitions of $H_Y$ and for the renormalization procedure giving a meaning to $\muren (\al) $ in (\ref{e:H}) and to the `strength' parameter $\al$, we refer to 
\cite{AFH79,AGH82,AGHH12} in the case $\al \in \RR$, and to 
$\cite{AGHS83,AK17}$ in the case $\al \not \in \RR$. 
Note that, in the case $\al \in  \RR$, 
the operator $H_Y $ is self-adjoint in $L^2_\CC (\RR^3)$; and 
in the case $\al \in  \CC_- $, 
$H_Y$ is closed and maximal dissipative (in the sense of \cite{E12}, or in the sense that $\ii H_Y$ is maximal accretive).
\end{rem}

In the future, we will consider the finite collection $Y$ of distinct points as a finite simple multiset  in $\RR^3$.

If $\# Y \ge 1$, the set of resonances $\Si (H_Y)$ of the deterministic operator $H_Y$ 
is by definition the set of zeroes of 
the determinant 
$
\det \Ga_Y (\cdot)  ,
$
which we call the  characteristic determinant.
The multiplicity of a resonance $k$ will be understood as the multiplicity 
of a corresponding zero of 
$\det \Ga_Y $, which is an analytic function in $z$ \cite{AH84,AGHH12}.
Equipped with the multiplicity of any resonance, the set $\Si (H_Y)$ becomes a multiset.

\begin{rem}
While it is widely assumed among specialists that the various definitions of (algebraic) multiplicities of resonances coincide (e.g., in \cite{AH84} and in \cite{DZ19} for $H_Y$; we point out that $H_Y$ can be easily placed into the black box formalism of \cite{SZ91}), but we do not know to what extent this assumption is actually checked.
In particular, the answer to this question for the multiplicity of a zero resonance depends on the way how this resonance is handled (in some of the works, $0$ is excluded from the set of resonances by definition). 
\end{rem}

The characteristic determinant $\det \Ga_Y (\cdot) $ is obviously an exponential polynomial.

Let us consider in more detail the structure of its simple modification 
\begin{gather} \label{e:D}
D_Y (z) := (-4 \pi)^{\#Y} \det \Ga_Y (z) ,
\end{gather}
which we call the modified characteristic determinant and which is also an exponential polynomial. 
The multiplication by $(-4 \pi)^{\#Y}$ 
will simplify the appearance of the formulae below.

The statement that $D_Y $ is an exponential polynomial means that it has the form 
\begin{equation} \label{e:CanForm}
 \sum_{j=1}^{\#B} P_{B_j} (z) e^{\ii B_j z} ,
\end{equation}
where $B = \{ B_j \}_{j=1}^{\#B} $ is a finite sequence of complex numbers with $\#B \in \NN$ and $P_{B_j} (\cdot)$ are polynomials. 
We will say that a function is an exp-monomial if it has 
the form $e^{\ii B_0 z} p (z)$ with $B_0 \in \CC$ and a nontrivial polynomial $p$ 
(nontrivial in the sense that $p (\cdot) \not \equiv 0$). 

Expanding by the Leibniz formula the determinant $\det \Ga_Y (z)$,
one sees that 
$D_Y (z) $
is a sum of terms of the form
\begin{equation} \label{e:Dterms}
e^{\ii z  V_\si (Y)} p^{\si,Y} (z)
\end{equation}
taken over all permutations $\si$ in the symmetric group $S_N$ with $N= \#Y$.
Here the numbers $V_\si (Y)$  and the polynomials $p^{\si,Y} (\cdot)$ 
have the form 
\begin{gather} \label{e:al si}
\textstyle V_\si (Y) :=  \sum_{j =1}^{\# Y}  |y_j - y_{\si (j)}| , \quad 
p^{\si,Y} (\zeta) := \ep_\si C_1 (\si,Y) \prod\limits_{j : \si (j) = j  } (\ii z- 4 \pi \al) , 
\end{gather}
where $C_1 (\si,Y) := \prod_{j : \si (j) \neq j  } |y_j - y_{\si (j)}|^{-1} $
(in the case $\si = \id$, we put $K_1 (\id,Y) := 1$), $\ep_\si$ is the permutation sign (the Levi-Civita symbol), and $\id$ is the identity permutation.

For the particular case of $D_Y$ as given by (\ref{e:D}), the sequence $ \{ B_j \}_{j=1}^{\#B}$ and polynomials $P_{B_j} $ in 
(\ref{e:CanForm}) can always be chosen such that
\begin{itemize}
\item the  sequence $\{ B_j \}_{j=1}^{\#B} $ consists of increasing nonnegative numbers and \item each of the polynomials $P_{B_j} (\cdot) $ is nontrivial. 
\end{itemize}
In such a case, we say that (\ref{e:CanForm}) is the \emph{canonical form} and that $B_j$ are the \emph{frequencies} of the exponential polynomial $D_Y$. Similarly, $V_\si (Y)$ is called the frequency of the exp-monomial (\ref{e:Dterms}) (in the terminology of \cite{AK20},  $V_\si (Y)$ is the metric length of the directed graph associated with $Y$ and a permutation $\si$).
Note that  we always have  
\begin{gather} \label{e:b=0}
\text{$B_{1} = 0$ and $P_0 (z) = P_{B_1} (z)=  (\ii z - 4 \pi \al)^{\#Y} $}.
\end{gather}

Consider now the random operator $H_\Up$, where 
$\Up$ is assumed to be a finite point process on $\RR^3$ satisfying (A0).
In particular, $\Up$ is a proper point process (see e.g. \cite{LP17}). Consequently,
there exist an $ \NN_0$-valued random variable $\nu$ and $\RR^3$-valued random variables
$\Up_j $, $j \in \NN$, such that $\Up $ can be considered a.s. as a finite collection of random $\RR^3$-points, $\Up = \{ \Up_j \}_{j=1}^\nu$ a.s..

In the case $\# Y = 0$, one has $H_Y = - \De$ and $\Si (H_Y) = \emptyset$, and so we put
$D_Y (z) := 1$

With the above deterministic definitions, the  random exponential polynomial 
$D_{\Up} (z)$ is now a.s. defined and generates the random multiset of resonances $\Si (H_{\Up (\om)})$  as the multiset of its zeros. 

\begin{ex} 
Consider a mixed binomial process $\Up$ with mixing distribution $\VV$ given by $\VV (\{j\}) = 1/3$ for $j=0,1,2$ and the standard multivariate normal  distribution $N (\mathbf{0},I_{\RR^3})$ as the sampling distribution $\QQ$.
That is, $\Up = \{\Up_j\}_{j=1}^{\nu}$,  where $\nu$, $\Up_1$, and $\Up_2$ are mutually independent random variables, $\Up_1$ and $\Up_2$ are normally distributed  with the law $N (\mathbf{0},I_{\RR^3})$, and $\PP \{\nu = j\} = 1/3$, $j=0,1,2$. Let us introduce the random variables $\ell = |\Up_1   - \Up_2 | $ and $\diam \Up$
(so $\diam \Up$ is equal to $\ell$ when $\nu =2$, and to $0$ when $\nu \le 1$).
Note that $\ell / \sqrt{2}$ has a $\chi_3$-distribution as its law.

We observe that for $\# \Up = \nu = 0$, we have $\Si (H_\Up) = \emptyset$. If $\# \Up = 1$,  then 
$\det \Ga_\Up (z) =  \al - \tfrac{\ii z}{4 \pi} $ and so 
$\Si (H_\Up)$ consists of one point $(- \ii) 4 \pi \al $ of multiplicity 1.
Each of the two aforementioned events has probability $1/3$. 
In the event 
$\{ \# \Up =2, \ \ell = 0\}$ having zero probability, the Hamiltonian $H_{\Up (\om)}$ (and so $\Si (H_{\Up (\om)})$) is formally not defined.

Assume now that the event $\om \in \{\# \Up =2, \ \ell > 0\}$ (with probabilty $1/3$) takes place.
Then the multiset $\Si (H_\Up)$ consists of zeroes 
of the exponential polynomial $D_{\Up (\om)} (z) = (\ii z - 4 \pi \al)^2 - \left(e^{\ii z \ell (\om)}/ \ell (\om)\right)^2$ and is a countable sequence with an accumulation point at $\infty$.  A more detailed  description of the set of zeros of this transcedental function of $z$ in terms of $\al$ and $\ell (\om)$ can be found in \cite{AH84,AGHH12,AK19} (see also Section \ref{s:AsPp}). 
\end{ex}

The above considerations easily lead to the following result.

\begin{thm} \label{t:pp} 
Assume that the point process $\Up$ satisfies (A0). Then:
\item[(i)] The random multiset of resonances $\Si (H_\Up)$ is a proper point process on $\CC$. 
\item[(ii)] 
\begin{gather} \label{e:NSi} 
\# \Si (H_{\Up(\om)} ) = \left\{ 
\begin{array}{ll} 
\infty, & \text{for $\om$ such that $\# \Up (\om)  \ge 2$ and $\Up_1 (\om) \neq \Up_2 (\om)$};\\
1, & \text{for $\om$ such that $\# \Up (\om) = 1$}; \\
0, & \text{for $\om$ such that $\# \Up (\om) = 0$}; 
\end{array} \right.
\quad \text{ a.s. }
\end{gather}
\end{thm}

\begin{proof}
\emph{(i)} Obviously, $D_{\Up(\om)} (\cdot) \not \equiv 0$  and $D_{\Up(\om)}$ is a random entire function for a.a. $\om \in \Om$. This implies that the multiset of zeros of $D_\Up $ (and so the multiset $\Si (H_\Up)$) is a locally finite and proper point process in $\CC$ (see \cite[pp. 338-340]{S12}; the scheme of the proof in the particular case of random polynomials can be found in  \cite{A66_Rps,BS86}). 

\emph{(ii)} The case $\# \Up (\om) = 0$ follows from the fact that $\Si (H_{\emptyset}) = \Si (-\De) = \emptyset$. 
Assume that $\# \Up (\om) = 1$. Then a direct computation gives that $\Si (H_{\Up (\om)})$ consists of one point $(- \ii)4 \pi \al $ of multiplicity 1. Finally, assume that $2 \le \# \Up (\om)  < \infty $ and $\Up (\om)$ is simple. Then it is easy to see that, in the Leibniz expansion of $D_{\Up (\om)}$, the exp-monomials (\ref{e:Dterms}) with positive frequencies cannot completely cancel each other \cite{AK17,AK19}. So, additionally to the zero frequency (\ref{e:b=0}), the exponential polynomial $D_{\Up (\om)} (\cdot)$  has at least one positive frequency. The existence of two different frequencies implies the existence of infinite number of zeroes of $D_{\Up (\om)}$ \cite{BC63,LL17}. 
\end{proof}

\section{Asymptotic density and Weyl-type asymptotics with probability 1\label{s:ad}}

A substantial part of the mathematical studies of deterministic resonances 
is devoted to the asymptotics of the their counting function 
\begin{align} \label{e:NR}
& \N_{H_{Y}} (R)  =  \# \{ k \in \Si (H_Y) \ : \ |k| \le R \} .
\end{align}
In \cite{LL17}, 
the asymptotics 
$\N_{H_{Y}} (R) = \frac{C}{\pi} R + O (1)$
as $R\to \infty$
with a certain constant $C \ge 0$
was established for deterministic  Hamiltonians $H_{Y}$ with $\# Y = N \in \NN$ point interactions 
and it was proved that $C \le V (Y) := \max_{\si \in S_{\# Y}} \sum_{j=1}^{\#Y} | Y_{j} -  Y_{\si (j)}| $. The number $V (Y)$
 was called in \cite{LL17} the \emph{size of the set $Y$}. In the case $C=V(Y)$, it was said (slightly changing the wording in \cite{LL17}) that the \emph{Weyl-type asymptotics} of $\N_{H_{Y}} (R)$ takes place. 

We use in the present paper the terminology of \cite{AK19} and say that $\Ad (H_Y) := C/\pi$ is the \emph{total asymptotic density of resonances of $H_Y$}.
 This is motivated by the equality (see (\ref{e:NR}))
\begin{equation} \label{e:limsup}
\Ad (H_Y) = \lim_{R \to \infty} \frac{\N_{H_Y} (R)}{R} .
\end{equation}

By Theorem \ref{t:pp}, the total asymptotic density of random resonances $\Ad (H_\Up)$ is an $[0,+\infty]$-valued random variable for any point process $\Up$ satisfying (A0). Combining this with the deterministic result of  \cite{LL17} one sees that $\Ad (H_\Up)$ is a $[0,+\infty)$-valued  random variable and that 
\[
\Ad (H_\Up)  \le \frac{V (\Up)}{\pi} \text{ a.s.}
\]

The main result of this section says, roughly speaking, that for  point processes $\Up$ with good enough `diffuse'  sampling distributions the Weyl-type asymptotics for random resonances of $H_\Up$ holds with the probability 1. For the sake of simplicity, we prove this result only for 
the uniform binomial processes $\Theta (m,\BB_r)$ in $\RR^3$-balls (see Section \ref{s:i}).

\begin{thm} \label{t:Ad=as}
If $\Up \in  \Theta (m,\BB_r)$ with $m \in \NN$, then a.s. we have $ \Ad (H_\Up) = V (\Up)/\pi$.
\end{thm}

We obtain this theorem in Section \ref{ss:proofAd} from the strengthened version of the deterministic result of \cite{AK20}. 

Namely,  for a deterministic $H_{Y}$ with $N \in \NN$ point interactions, it follows from \cite{AK20} that the Weyl-type asymptotics is generic in the sense described below. We consider $Y=\{Y_j\}_{j=1}^N$ as an $N$-tuple of $\RR^3$ vectors and identify it with a vector in the space $(\RR^3)^N=\RR^{3N}$ with the standard $\ell^2$-metric. The assumption that the interaction centers $Y_j$ are distinct means that $Y$ belongs to 
the family $\AAA$ of admissible $N$-tuples that is defined by 
\[
\AAA := \{ Y \in (\RR^3)^N \ : \ Y_j \neq Y_{j'} \ \text{ for } \ j \neq j' \} 
\]
and that is considered as an induced metric space and an induced measurable space with (3N-dimensional) Lebesgue measure.
Then \cite{AK20} implies that the set of $Y$ such that $\N_{H_Y}$ has a Weyl-type asymptotics is nowhere dense. This result is not enough to prove Theorem \ref{t:Ad=as} (because there exist nowehere dense subsets of $\RR^{3N}$ with positive Lebesgue measure).

Recall that  $S$ is a proper analytic subset of an open set $O \subset \RR^d$ if there exists a real analytic function $f$ on $O$ such that $f \not \equiv 0$ on $O$ and $S=\{ x \in O : f (x) = 0 \}$.

We prove in Section \ref{ss:proofAd} the following strengthening of the aforementioned result of \cite{AK20}.

\begin{thm} \label{t:Admeas}
Let the set $\AAA_0 \subset (\RR^3)^N$ consist of all admissible $N$-tuples $Y \in \AAA$ so that $\N_{H_Y}$ has non-Weyl-type asymptotics (i.e., so that $\Ad (H_Y) < V (Y)/\pi$). Then:
\item[(i)] $\AAA_0$ is a subset of a certain proper analytic subset of $\AAA$;
\item[(ii)] $\AAA_0$ is a set of zero Lebesgue measure. 
\end{thm}

\begin{rem}
Lojasiewicz's theory about the structure of analytic varieties \cite{L91} yields much stronger restrictions (than those of Theorem \ref{t:Admeas} (ii)) on  the proper analytic subset containing the non-Weyl-type set $\AAA_0$. We refer to \cite{KP02} for the discussion of the theory of real analytic sets.
\end{rem}

\subsection{Proofs of Theorem \ref{t:Ad=as} and \ref{t:Admeas}\label{ss:proofAd}}

The equivalence classes of edge-equivalent permutations 
$\si \in S_N$ for the $N$-tuple $Y$ were introduced in \cite{AK20}. This definition was given in terms of directed and undirected graphs associated with permutations. We refer to \cite[Section 3]{AK20} for the details and would like to notice here that the following fact was also proved: two permutations $\si, \si' \in S_N$ are  edge-equivalent if and only if $V_\si (Y) = V_{\si'} (Y)$ for every $Y \in \AAA$. We recall that $V_\si$ is defined in (\ref{e:al si}) 
(it is the the metric length of the aforementioned directed metric graph associated with $\si$ and $Y$).

Let us denote by $\wt n \in \NN$ the number of 
edge-equivalence classes in $S_N$ and let us take one representative $\wt \si_j$, $j=1,\dots, \wt n$, in each of them.
We use also the following observation of \cite{AK20}:
if $Y$ belongs to the set
\[
\AAA_1 := \{ Y \in \AAA \ : \ \V_{\wt \si_j} (Y) \neq \V_{\wt \si_m} (Y) \text{  if } j \neq m \} ,
\]
then there is no cancellation of the exp-monomials (\ref{e:Dterms}) with the lowest possible frequency $(-V(Y))$ after the summation of (\ref{e:Dterms}) required by the Leibniz formula for $\det \Ga_{Y} $. Thus, for every $Y \in \AAA_1$ the Weyl-type asymptotics takes place.

\begin{proof}[Proof of Theorem \ref{t:Admeas}]
For each permutation $\si$, the function $\V_\si (\cdot)$ is a sum of terms of the form
\[
|Y_j  - Y_{j'} |= \left(\sum_{m=1}^3 [Y_{j,m} - Y_{j',m} ]^2 \right)^{1/2} ,
\]
where $j' = \si (j)$ and where $Y_{j,m}$, $m=1,2,3$,
are the $\RR^3$-coordinates of $Y_j$, $1 \le j \le N$.
Therefore, $\V_\si (\cdot)$ is a real analytic function in the variables $Y_{j,m}$ ($1 \le j \le N$, $m=1,2,3$) on $\AAA$.

Let us take now the representatives $\wt \si_j$, $j=1,\dots, \wt n$, of edge-equivalent classes of permutations, which were described above. We see that the function $f_{j,m} (Y) = \V_{\wt \si_j} (Y) - \V_{\wt \si_m} (Y) $ is real analytic function on $\AAA$. Moreover, if $j \neq m$ this function is not trivial on $\AAA$, and so the set $\AAA_0^{j,m} := \{ Y \in \AAA : f_{j,m} (Y) = 0 \}$ of its zeroes is a proper analytic subset of $\AAA$. 

We will use the following well-known fact (see e.g. \cite{KP02}): 
\begin{gather} \label{e:meas=0}
 \text{a proper analytic subset of an open set in $\RR^d$ has measure zero.}
\end{gather}
Thus, each of the sets  $\AAA_0^{j,m}$ with $j \neq m$ is of measure zero and so is 
their union $\wt \AAA_0 = \bigcup\limits_{1 \le j < m \le \wt n } \AAA_0^{j,m}$. This union 
is obviously also a proper analytic subset of $\AAA$ (it is the set of zeros of the  function $f (Y) = \prod\limits_{1 \le j < m \le \, \wt n } f_{j,m}$).

As it was mentioned above, the results of \cite{AK20} imply that, if $Y \in \AAA \setminus \wt \AAA_0$, the Weyl-type asymptotics takes place. Summarizing, we see that $\AAA_0\subset \wt \AAA_0$, and that $\AAA_0$ is a proper analytic subset of $\AAA$ and it has measure zero. This completes the proof.
\end{proof} 

\begin{proof}[Proof of Theorem \ref{t:Ad=as}]
Assume that $\Up$ is a binomial process of the sample size $m \ge 1$ with the  sampling distribution uniform in $\BB_r$, i.e., 
$\Up \in \Theta (m,\BB_r)$. That is,  $\Up = \{ \Up_j \}_{j=1}^m$, where $\RR^3$-valued random variables $\Up_j$ are i.i.d. and uniformly distributed in $\BB_r$. We can consider $\Up$ as a random vector in $(\RR^3)^m$. The distribution of $\Up$ has the density with respect to Lebesgue's measure over $\RR^{3N}$. Integrating this density over the  measure zero set $\AAA_0$ (w.r.t. the Lebesgue measue), we see from Theorem \ref{t:Admeas} that the probability of the event $\{ \om : \Ad (\Up) < V (\Up)/\pi\}$ equals 0. This completes the proof.
\end{proof}

\section{Point process describing the asymptotics of random resonances\label{s:AsPp}}

The goal of this section is to show that the structure of the set of  random resonances of $H_\Up$ near $\infty$ can be described by a point process on $\RR_+$.

Consider first 
\begin{gather} \label{e:Y0}
\text{a deterministic collection $Y\subset \RR^3$ such that $Y$ is simple and $ 2 \le \# Y < \infty$.}
\end{gather}
Then the multiset of resonances $\Si (H_Y)$ has the global structure of a finite number of  sequences going to $\infty$ with prescribed asymptotics \cite{AK19}.
Namely, there exists a   sequence $\{\Kp_j (Y) \}_{j=1}^{n_1 (Y)} $ of $n_1 (Y) \in \NN$ positive numbers such that 
the multiset $\Si (H_Y)$ is essentially the union $\bigcup_{j=1}^{N} \{k_{j,m}\}_{m \in \ZZ }$ of the sequences satisfying 
\begin{gather} \label{e:kjm}
k_{j,m} = 
  \pi \Kp_j (Y) (2 m +1) - \ii \Kp_j (Y) \Ln |\pi \Kp_j (Y) (2 m +1) | +  O (1)  \text{ as $|m| \to \infty$},
\end{gather}
$j=1, \dots, n_1 (Y)$.
(In \cite{AK19} a more precise asymptotic formula is given, but we do not need it in the present paper.)
`Essentially' in this context means that  one multiset can be obtained from the other by possible addition or exclusion of a finite number of elements.

The collection $\Kp (Y)=\{\Kp_j (Y) \}_{j=1}^{n_1 (Y)}$ of the leading parameters of the asymptotic sequences (\ref{e:kjm}) can be considered as a multiset and 
we assume that $\Kp_j$ are ordered such that 
\[
\Kp_1 (Y) \le \Kp_2 (Y) \le \dots \le \Kp_{n_1 (Y)} (Y) .
\]
Note that some of its elements are actually multiple. Namely, the following 
facts were proved in \cite{AK19} under condition (\ref{e:Y0}):
\begin{align} \label{e:2<mK<mY} 
& 2 \le \ n_1 (Y) \ = \ \# \Kp (Y) \ \le \ \# Y ,
\\ 
& \Kp_1 (Y) = \Kp_2 (Y) = 1/\diam Y \label{e:K1diam}
\end{align}
(the latter means that the multiplicity $\mult (\Kp_1 (Y)) $ of the minimal parameter  $\Kp_1 (Y)$ is at least $2$).

Since $\# \Si (H_Y) \le 1$ if $\# Y \le 1$, it is logical to put 
\[
\text{$\Kp (Y) = \emptyset$ in the cases $\# Y =0$ and $\# Y =1$.}
\]

Then the random multiset $\Kp (\Up)$ of the parameters of the asymptotic sequences is a.s. defined. Let us show its mesurability with respect to the
probabilistic $\sigma$-algebra $\F$ of the underlying probability space.

\begin{thm} \label{t:Kpp}
Assume (A0). Then
$\Kp (\Up)$ is a finite proper point process on $\RR_+$.
\end{thm}

\begin{proof}
It follows from definition of $K_j (Y)$ that the random set $K(\Up)$ lies in $\RR_+$ whenever it is defined and is nonempty.
Since by (\ref{e:2<mK<mY}) we have $ \# K (\Up) \le \# Y < \infty$ a.s., it is enough to prove that 
$\int_\RR f  \dd \eta_{_{\Kp (\Up)}}$ is a random variable for every $f (\cdot)$ from the space $C_{\comp} (\RR)$ of $\RR$-valued compactly supported continuous functions on $\RR$ (recall that 
$\eta_{_{\Kp (\Up (\om))}}$ is the counting measure for $\Kp (\Up (\om))$).

For $R>0$ and a finite simple deterministic $Y$, consider the functional 
\[
J_{Y,R} (f) = \frac{\pi}{R} \int_{\ii \CC_+ \cap \DD_R } f \left(  \frac{- \im z}{\Ln (\re z+1)}\right) \  \frac{-\im z}{\Ln (\re z+1)} \ \dd \eta_{_{\Si (H_Y)}}
\]
defined on $C_{\comp} (\RR)$.
By Theorem \ref{t:pp}, $\Si (H_\Up)$ is a point process, and so 
the stochastic version $J_{\Up,R} (f)$ of the above functional is an $\RR$-valued random variable for each function $f \in C_{\comp} (\RR)$. Consequently,
$\lim_{R \to +\infty} J_{\Up,R} (f)$ is an $[-\infty,+\infty]$-valued random variable.

If $Y$ satisfies (\ref{e:Y0}), the asymptotic formulae  
(\ref{e:kjm}) imply that 
\[
\lim_{R \to +\infty} J_{Y,R} (f) = \int_{\RR} f \dd \eta_{_{K (Y)}} = \int_{\RR_+} f \dd \eta_{_{K (Y)}} \in \RR.
\]
If $\# Y \le 1$, the above limit equals $0$.
Thus, $\int_{\RR} f \dd \eta_{K (\Up)}$ is an $\RR$-valued random variable for every $f \in C_{\comp} (\RR)$. This completes the proof.
\end{proof}


\begin{cor} \label{c:Kj} Assume that $\Up$ satisfies (A0). Then:
\item[(i)] There exist an $\NN_0$-valued random number $n_1 = n_1 (\Up)$ and 
a sequence of $(-\infty,+\infty] $-valued random variables $\K_j$, $j \in \NN$, such that a.s.
$\Kp (\Up) = \{ \K_j \}_{j=1}^{n_1 (\Up)}$ and $\K_j \le \K_{j+1}$, $j \in \NN$.
\item[(ii)] The random variables $\K_1$ and $\K_2$ can be chosen in (i) so that 
$\K_1 = \K_2 = 1/\diam \Up$ a.s.
(here the convention about the value $1/0$ is not important, but it  can be put for simplicity to be $+ \infty$).
\end{cor}

\begin{proof} (i) follows from Theorem \ref{t:Kpp}. (ii) follow from (\ref{e:K1diam}).
\end{proof}

In what follows it is assumed that for every $\Up$, we fix a certain sequence $\K_j$, $j \in \NN$, of random variables satisfying Corollary \ref{c:Kj} (i)-(ii) and denote it by $\K_j (\Up)$.
(The reason for this convention is that we would like to deal with the maximal and the minimal of the parameters $\{ K_j (\Up) \}_{j=1}^{\# K (\Up)}$ and redefine them avoiding the situation when they do not exist on the event $\{ \om : \# \Up < 2 \}$ of possibly positive probability.)

We pay special attention to the minimal $\Kp_1 (\Up)$ of the parameters $\Kp_j (\Up)$ because $\Kp_1 (\Up)$ corresponds to one or several sequences having asymptotically the smallest possible values for the resonance's minus imaginary part $|\im k|$. The parameter $|\im k|$ can be interpreted as the exponential decay rate of monochromatic oscillations in the settings of the acoustic-type wave equation (see e.g. \cite{CZ95,DZ19}). In the context of the  Schrödinger equation, $\Ga (k) = 4 |\im (k) \re (k) |$ is called the width of the resonance \cite{RSIV}. Resonances $k$ with small values of $|\im k|$ are considered to be `narrow resonances' and are more visible in physical  scattering experiments, see e.g.  \cite{E84,RSIV}.

\section{Limits of random asymptotic structures under growing intensity\label{s:limits}}

As it is shown in Section \ref{s:AsPp}, the asymptocial behaviour of  random resonances at $\infty$ is described by the  finite point process 
$K (\Up) = \{ K_j (\Up)\}_{j=1}^{\# K(\Up)}$ on $\RR$.
This naturally poses a question about the asymptotics 
of the random counting measures $\eta_{_{K(\Up^{[m]})}}$ for a reasonably chosen  sequence of point processes $\Up^{[m]}$.
It makes sense to assume that with the growth of $m \to + \infty$ either the `intensity' or the support of $\Up^{[m]}$ grows unboundedly. 

As a simple example of such a reasonable sequence of point processes, one can take uniform binomial processes $ \Up^{[m]} \in  \Theta (m,\BB_{r})$ in a fixed 3-D ball $\BB_r$ with the total intensity $m$ going to $+\infty$. That is for each $m$ there exists a sequence $\{\xi^{[m]}_{j}\}_{j=1}^{m}$ of independent random variables with uniform distribution in the unit ball $\BB_r$ such that $\Up^{[m]} = \{ \xi^{[m]}_j \}_{j=1}^{m}$.
 Formally, the elements $\xi_j^{[m]}$ of the sequence 
 $\Up^{[m]}$ depend on $m$. However, this dependence can be often neglected. There exists an infinite sequence $\{\xi_{j}\}_{j=1}^{\infty}$ of uniformly distributed in $\BB_r$ i.i.d. random variables such that each of the point processes 
 \begin{gather}
 \wt \Up^{[m]} = \{ \xi_j \}_{j=1}^{m}, \qquad m \in \NN,
 \label{e:tUp}
 \end{gather}
 has the same distribution as $\Up^{[m]}$. For all the purposes of the present paper, $\{ \Up^{[m]} \}_{1}^{\infty}$ can be replaced by $\{\wt \Up^{[m]} \}_{1}^{\infty}$ (the only reason for this replacement is to simplify the notation).

The first questions in connection with the limiting behavior of 
$K (\Up^{[m]})$ concern the limits of the random variables 
\begin{multline} \label{e:Kminmax}
n_{1}^{[m]} := n_1 (\Up^{[m]}) = \#  K (\Up^{[m]}), \quad \K_{\min}^{[m]} := \K_1 (\Up^{[m]}) , 
\text{ and }
\qquad \K_{\max}^{[m]} := \K_{n_{1}^{[m]}}^{[m]} (\Up^{[m]}) .
\end{multline}

Another interesting limiting behavior question concerns the total asymptotic densities $\Ad (H_{\Up^{[m]}})$ (cf. the introduction to \cite{S14}). 
Recall that $V (Y) := \max_{\si \in S_{\# Y}} \sum_{j=1}^{\#Y} | Y_{j} -  Y_{\si (j)}| $ is called the size of the set $Y$ (see Section \ref{s:ad}).

\begin{cor} \label{c:Kprep}
Let $m \ge 2$ and let $\Up^{[m]} =  \{ \xi_j \}_{j=1}^m$ be a collection of $m$ independent uniformly distributed in $\BB_r$ for some $r>0$ random points $\xi_j$. 
Then:
\item[(i)] $ \#  K (\Up^{[m]}) = m $ a.s.,
\item[(ii)] $\K_{\min}^{[m]}$ is equal in distribution to the random variable $\displaystyle \frac{1}{ \max\limits_{1 \le i,j\le m} |\xi_j - \xi_i|}$. In particular, $\K_{\min}^{[m]} \ge \frac{1}{2r} $ a.s.
\item[(iii)] \[ \K_{\max}^{[m]} \ge \frac{m}{\max\limits_{\si \in S_m} \sum_{j=1}^{m} | \xi_{j} - \xi_{\si (j)}|} = \frac{m}{\pi \Ad (H_{\Up^{[m]}})} \quad \text{ a.s. }\]
\end{cor} 

\begin{proof}
(i) follows from Theorem \ref{t:Ad=as} and the proof of the statement (iv) of \cite[Theorem 3.4]{AK19}.
Statement (ii) follows from Corollary \ref{c:Kj} (ii). Statement (iii) can be easily obtained from the combination of Theorem \ref{t:Ad=as}  with the convexity of the distribution diagram for the zeros of the characteristic determinant $\det \Ga_{\Up^{[m]}} (\cdot) $
(see \cite[Theorem 3.4 and Sect. 3.1]{AK19}).
\end{proof}

\subsection{Limit law for the `most narrow' asymptotic sequence}

Consider now the limit of the random variables $\K_{\min}^{[m]}$ as
$m \to \infty$.

\begin{thm}
Consider a sequence of uniform binomial processes $ \Up^{[m]} \in  \Theta (m,\BB_{r})$, $m \in \NN$. Then:
\item[(i)] As $m \to +\infty$, we have $\K_{\min}^{[m]} \to \frac{1}{2r}$ in probability.
\item[(ii)] The (rescaled) limit distribution of the  random variable $\K_{\min}^{[m]} - \frac{1}{2r}$ is given by 
\begin{equation}
\PP \left\{   
m^{2/3} \left(\K_{\min}^{[m]} - \frac{1}{2r} \right) \le t 
\right\} \to 1 - \ee^{- 48 r^3 t^3} \text{\quad as $m \to \infty$ }, \quad t>0
\label{e:KminLim}
\end{equation}
(recall that $\K_{\min}^{[m]} - \frac{1}{2r} \ge 0 $ a.s.).
\end{thm}

\begin{proof}
(i) is obvious from Corollary \ref{c:Kprep} (ii), i.e., from the fact that $\K_{\min}^{[m]}$ 
is equal in distribution to the random variable $\displaystyle \frac{1}{ \max\limits_{1 \le i,j\le m} |\xi_j - \xi_i|}$.
Combining this fact with the result of \cite[Theorem 1.1]{MM07}
on the limiting distribution of $\max\limits_{1 \le i,j\le m} |\xi_j - \xi_i|$, one obtains
statement (ii).

\end{proof}

\subsection{Estimates on the growth of 
the total asymptotic density\label{ss:LLAd}
}

The study of the limit law as $m \to \infty$ for the maximal leading parameter $\K_{\max}^{[m]}$ is a more difficult problem.
This parameter is connected with the total asymptotic density of resonances $\Ad (H_{\Up^{[m]}})$ and so, due to Theorem \ref{t:Ad=as}, with the size $V (\Up^{[m]})$ of the random set $\Up^{[m]}$. A simple version of this connection is given by the inequality $\displaystyle \K_{\max}^{[m]} \ge \frac{m}{V (\Up^{[m]})}$ (see Corollary \ref{c:Kprep} (iii)). A more precise dependence in the deterministic case can be seen from \cite[formula (3.6)]{AK19}.

The following theorem describes the rate of grow as $m \to \infty$ of  the total asymptotic densities $\Ad (H_{\Up^{[m]}})$ and of the sizes 
$V (\Up^{[m]})$, which according to Theorem \ref{t:Ad=as} are connected by $\Ad (H_{\Up^{[m]}}) = \frac{V (\Up^{[m]})}{\pi}$ a.s..

\begin{thm} \label{t:Vlim}
Let $r>0$ and $ \Up^{[m]} \in  \Theta (m,\BB_{r})$, $m \in \NN$.
Then 
\begin{gather} \label{e:Vgrow}
\liminf_{m\to \infty}\PP \left\{ \frac{V (\Up^{[m]})}{r} >\frac{36}{35}m + \frac{2\sqrt{87}}{35} t\sqrt{m} \right\} \ge
1  - \Phi (t),
\end{gather}
where $\Phi (t) = (2\pi)^{-1/2} \int_{-\infty}^t e^{-s^2/2} \dd s$ (the standard normal distribution function). In particular, the following estimate is valid for the asymptotic density $\Ad (H_{\Up^{[m]}})$ of resonances
\begin{gather} \label{e:Vgrow2}
\lim_{m\to \infty}\PP \left\{ \Ad (H_{\Up^{[m]}})> \frac{mr}{\pi} \right\} \to 1 
\ \text{ as $m \to \infty$}
\end{gather}
\end{thm}
\begin{proof}
For convenience of the notation, we replace  each process $\Up^{[m]}$ by the process $\wt \Up^{[m]} = \{ \xi_j \}_{j=1}^m$ defined in (\ref{e:tUp}). This does not influence the estimates below.

 Let 
\[
\text{$m_* = 2 \lfloor m/2 \rfloor$, i.e., $m_* = m$ if $m$ is even, and to $m_* = m-1$ if $m$ is odd.}
\]
Then, from the definition of $V (\cdot)$, we have
\[
V (\wt \Up^{[m]}) \ge 2 S_{m_*/2}, \quad \text{ where } S_{m} = \sum_{j=1}^{m} |\xi_{2j-1} - \xi_{2j} |.
\]
The $\RR_+$-valued random variables $\la_j := \frac{|\xi_{2j-1} - \xi_{2j}|}{2r}$ are i.i.d. with the first two moments given by
\[
\EE (\la_1)  = 18/35 , \qquad \EE (\la_1^2) = 3/10 \qquad \text{(see \cite{H50} for the general formula).}
\]
Hence, the variance of $\la_j$ is $\Var \la_j =
\left( \frac{\sqrt{87}}{35 \sqrt{2}}
\right)^2 $.
Applying the Central Limit Theorem, we get 
\[
\PP \left\{  \frac{S_m}{2r} -  \frac{18}{35} m \ \le \ t \sqrt{m} \frac{\sqrt{87}}{35 \sqrt{2}} \right\} \to \Phi (t)
\] 
as $m \to \infty$.
This implies (\ref{e:Vgrow}) and, in turn, (\ref{e:Vgrow2}).
\end{proof}

\vspace{1ex}
\noindent
\textbf{Acknowledgements.} 
IK was supported by the VolkswagenStiftung project “Modeling, Analysis, and Approximation 
Theory toward applications in tomography and inverse problems”.
IK is grateful J\"urgen Prestin for the hospitality of the University of L{\"u}beck and to Baris Evren Ugurcan and the Hausdorff Research Institute for Mathematics of the University of Bonn for the possibility to participate in the activities of the trimester program  ``Randomness, PDEs and Nonlinear Fluctuations''.

\end{document}